\documentclass[a4paper,UKenglish,cleveref, autoref, thm-restate]{lipics-v2021}

\usepackage[colorinlistoftodos]{todonotes}
\usepackage[normalem]{ulem}

\usepackage[algoruled,linesnumbered]{algorithm2e/algorithm2e}
\usepackage{tikz-network}
\usepackage{subcaption}

\pdfoutput=1 %uncomment to ensure pdflatex processing (mandatatory e.g. to submit to arXiv)
\hideLIPIcs  %uncomment to remove references to LIPIcs series (logo, DOI, ...), e.g. when preparing a pre-final version to be uploaded to arXiv or another public repository

\usepackage{natbib}
\title{Improving polynomial bounds for the Graphical Traveling Salesman Problem with release dates on paths}

\titlerunning{Improving polynomial bounds for the GTSP-rd on paths} %TODO optional, please use if title is longer than one line

\author{Thailsson Clementino}{Instituto de Computação - UFAM, Brazil}{thailsson.clementino@icomp.ufam.edu.br}{}{}%TODO mandatory, please use full name; only 1 author per \author macro; first two parameters are mandatory, other parameters can be empty. Please provide at least the name of the affiliation and the country. The full address is optional. Use additional curly braces to indicate the correct name splitting when the last name consists of multiple name parts.

\author{Rosiane de Freitas}{Instituto de Computação - UFAM, Brazil}{rosiane@icomp.ufam.edu.br}{}{}

\authorrunning{T. Clementino and R. de Freitas} %TODO mandatory. First: Use abbreviated first/middle names. 

\Copyright{Thailsson Clementino de Andrade and Rosiane de Freitas Rodrigues} %TODO mandatory, please use full first names. LIPIcs license is "CC-BY";  http://creativecommons.org/licenses/by/3.0/


 %TODO mandatory; please add comma-separated list of keywords

\category{} %optional, e.g. invited paper

\acknowledgements{
    This research was partially supported by the Coordination for the Improvement of Higher Education Personnel - Brazil (CAPES-PROEX) - Funding Code 001, the National Council for Scientific and Technological Development (CNPq), and the Amazonas State Research Support Foundation - FAPEAM - through the POSGRAD 2024-2025 project. Also, under Brazilian Federal Law No. 8,387/1991, Motorola Mobility partially sponsored this research through the SWPERFI Research, Development, and Technological Innovation Project on intelligent software performance and through agreement No. 004/2021, signed with UFAM.
    The authors are part of the Algorithms, Optimization, and Computational Complexity (ALGOX) CNPq research group from the Postgraduate Program in Computer Science (PPGI), IComp/UFAM.
}
\nolinenumbers %uncomment to disable line numbering

\EventEditors{John Q. Open and Joan R. Access}
\EventNoEds{2}
\EventLongTitle{42nd Conference on Very Important Topics (CVIT 2016)}
\EventShortTitle{CVIT 2016}
\EventAcronym{CVIT}
\EventYear{2016}
\EventDate{December 24--27, 2016}
\EventLocation{Little Whinging, United Kingdom}
\EventLogo{}
\SeriesVolume{42}
\ArticleNo{23}
\begin{document}

\maketitle

\begin{abstract}
    The Graphical Traveling Salesman Problem with release dates (GTSP-rd) is a variation of the TSP-rd where each vertex in a 
    weighted graph $G$ must be visited at least once, respecting the release date restriction. The edges may be traversed 
    multiple times if necessary, as in some sparse graphs.  
     This paper focuses on solving the GTSP-rd in paths. We consider two objective functions: minimizing the route completion
      time (GTSP-rd (time)) and minimizing the total distance traveled (GTSP-rd (distance)). We present improvements to 
      existing dynamic programming algorithms, offering an $O(n)$ solution for paths where the depot is located at the 
      extremity and an $O(n^2)$ solution for paths where the depot is located anywhere. For the GTSP-rd (distance), we 
      propose an $O(n \log \log n)$ solution for the case with the depot at the extremity and an $O(n^2 \log \log n)$ solution for the
       general case.
\end{abstract}

\section{Introduction}

The Traveling Salesman Problem (TSP) is a well-known combinatorial optimization problem that
 seeks to determine the shortest possible route to visit a given set of cities exactly
  once and 
 return to the origin city \citep{cook2011traveling}. In the literature, the TSP is typically modeled as a weighted complete
  graph $G = (V,E)$, where each vertex in $V$ represents a city, and the weight associated with each edge in $E$ represents
   the distance between two cities. However, some works \citep{miliotis1981computational, ratliff1983order} explore the TSP
    without the assumption that the input graph is complete or without transforming it into a complete graph
     \citep{hargrave1962relation}. This variant of the TSP is referred as the Graphical Traveling Salesman Problem (GTSP).

  In the GTSP it is assumed that all cities (or vertices) are ready to be visited by the salesman at any time, but this 
  assumption may not align with real-world scenarios where we can view the salesman problem as a delivery problem and the
    goods or products become available at different times. 
To address these constraints, the Graphical Traveling Salesman Problem with Release Dates (GTSP-rd) was introduced as a
 variant of the problem. Moreover, in this variant  
 we define the starting vertex as the depot and allow more than one route starting and ending at the depot.
  The decision to be made is whether it is better to start a route that delivers the already available products to the 
  customers or wait until more products become available.
     
    In this paper, we address the GTSP-rd, focusing on instances where the inputs are 
paths. Our study explores the GTSP-rd with two 
different objective functions in this 
context: minimizing the route completion time (GTSP-rd (time)) and minimizing the total traveled 
distance (GTSP-rd (distance)).

Previous works in the literature \citep{archetti2015complexity, reyes2018complexity} have used dynamic programming to solve
 the GTSP-rd(time) and GTSP-rd(distance) problems. For GTSP-rd(time), these studies proposed an \( O(n^2) \) algorithm for
  paths with depots located at the extremities and an \( O(n^3) \) algorithm for more general path structures, where depots
   can be positioned anywhere. Similarly, for GTSP-rd(distance), algorithms with the same complexities were proposed.

In this work, we present improvements to the existing dynamic programming algorithms for GTSP-rd(time), including an
 \( O(n) \) solution for paths with depots at the extremities and an \( O(n^2) \) solution for more general path structures
  with depot located arbitrarily in any vertex. We also improve the GTSP-rd(distance) algorithms, proposing an \( O(n \log \log n) \) solution
   for the first case and an \( O(n^2 \log \log n) \) solution for the second.

The remainder of this paper is structured as follows: in Section~\ref{sec:tsprd}, we provide a formal definition of 
Graphical Traveling Salesman Problem with release dates (GTPS-rd). In Section~\ref{sec:on_path} we 
discuss this problem restricted to paths. We continue this discussion in Section \ref{subsec:path_depot_end} 
by examining a special case where the 
depot is situated at the extremity of the path. Following that, Section \ref{subsec:path_depot_anywhere}, we address the 
more general scenario of a path with the depot located anywhere. Finally, in Section \ref{sec:conclusion}
we present our concluding remarks and future works.  

\section{The Graphical Traveling Salesman Problem with release dates}\label{sec:tsprd}

Although previous and recently published works propose solutions for specific graph classes 
in the TSP-rd, they define (model) the problem as a complete graph. This results in a mismatch between the problem 
definition and the proposed solutions. To encompass potential solutions for specific graph classes, we define the 
problem considering not only complete graph as input. A similar approach was taken in the creation of the 
Graphical TSP \citep{fonlupt1993dynamic, cornuejols1985traveling, carr2023new}.

In this section, we formally define the Graphical Traveling Salesman Problem with release dates
 (GTSP-rd). The following definition enables constructing solutions without requiring the transformation
  of every input graph into a complete graph. Consequently, it allows us to exploit the inherent
   graph structure for more efficient solutions if they exist.

The Graphical Traveling Salesman Problem with release dates (GTSP-rd) can be defined as follows:
Given a simple connected graph \( G = (V, E) \), where the vertex set is the union of two sets, \( V = \{0\} \cup N \). The vertex \( 0 \) denotes the 
initial vertex (\textit{depot}), while the set of vertices \( N = \{1, \ldots, n\} \) 
represents the set of \textit{customers} to be visited. Each edge $(i,j) \in E$ is associated with a travel time (distance), denoted by \(d_{ij}\). 
Additionally, a release date $r_i \geq 0$ is associated with each vertex $i \in N$, indicating the earliest moment when the item to
 be delivered at vertex $i$ can depart from the depot.

A \textit{route} $\mathcal{R}$ is a closed walk in $G$ that starts and ends at the depot. 
Formally, $\mathcal{R} = [v_0, v_1, \dots, v_s, v_{s+1}]$, where $v_0 = v_{s+1} = 0$, 
\( S = \{v_1, \dots, v_s\} \subseteq N \), and \( (v_k, v_{k+1}) \in E \) for all \( k \in \{0, 1, \dots, s\} \). 
The vertices in \( S \) are partitioned into two subsets: \( S^d \), which contains the vertices where deliveries 
are made, referred to as \textit{delivery vertices}, and \( S^t = S \setminus S^d \), referred to as 
\textit{traverse vertices}.
The total distance traveled on a route is determined by
    \( d_{\mathcal{R}} = \sum_{0}^{s}d_{(v_k, v_{k+1})} \).
    The dispatch time of a route, \( T_\mathcal{R} \), is defined as the moment the salesman departs from the depot to 
    serve the set \( S^d \). The route \( \mathcal{R} \) must only begin after the latest release date in \( S^d \),
     ensuring \( T_\mathcal{R} \geq \max_{v \in S^d} \{r_v\} \).

A \textit{solution} to GTSP-rd consists of a sequence of $x$ routes $\mathcal{R}_1, \mathcal{R}_2, \dots, \mathcal{R}_x$ 
containing the vertices set $S_1,S_2, \dots, S_x$, these routes must be done consecutively by the Traveling Salesman in 
order of dispatch time $ T_{\mathcal{R}_{1}} \leq  T_{\mathcal{R}_{2}} \leq \dots \leq T_{\mathcal{R}_{x}}$. A route $\mathcal{R}_{j}$ 
can only leave the depot if the previous route has already been attended, that is,
 $T_{\mathcal{R}_{j-1}}+d_{\mathcal{R}_{j-1}} \leq T_{\mathcal{R}_{j}}$ for $j \in \{1,\cdots, x\}$. 
A solution to GTSP-rd is feasible if all the set $S_{j}^{d} \subseteq S_j$ form a partition of $N$.
  
  Figure \ref{fig:solution} provides an example of solution containing three routes. 
  $ \mathcal{R}_1 = [0, 4, 8, 9, 3, 2, 7, 5, 0]$ (green), 
    $ \mathcal{R}_2 = [0, 1, 5, 6, 5, 0]$ (orange) and 
    $ \mathcal{R}_3 = [0, 9, 10, 9, 0]$ (red) with the delivery vertices $S_1^d = \{4,8,9,3,2,7\}$, $S_2^d = \{1,5,6\}$
    and $S_3^d = \{10\}$. The dispatch times could be $T_{\mathcal{R}_1} = 5$, $T_{\mathcal{R}_2} = 22$ and
      $T_{\mathcal{R}_3} = 43$. Route $\mathcal{R}_1$ leaves the vertex $0$ in time $5$ and complete at time $22$ when
       route $\mathcal{R}_2$ can start. 
    A solution of GTSP-rd consists of one or more routes, where, by definition, at least the vertex $0$ is repeated in 
each route. Moreover, in some instances a vertex must be revisited several times.

  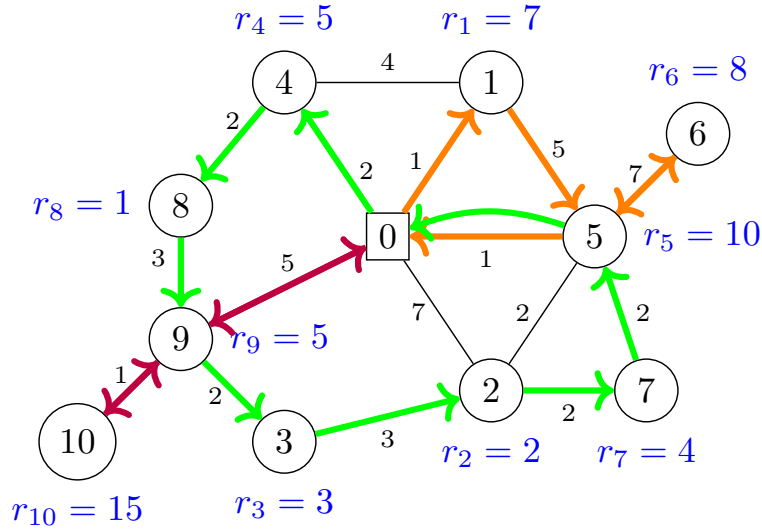
\begin{figure}[h]
    \begin{center}
        
        \resizebox{0.75\textwidth}{!}{
    \begin{tikzpicture}
      
      \node[draw] (0) at (0,0) {0};
      \node[draw,circle] (1) at (1,1.5) {1};
      \node[draw,circle] (2) at (1,-1.5) {2};
      \node[draw,circle] (3) at (-1,-2) {3};
      \node[draw,circle] (4) at (-1,1.5) {4};
      \node[draw,circle] (5) at (2,0) {5};
      \node[draw,circle] (6) at (3,1) {6};
      \node[draw,circle] (7) at (2.5,-1.5) {7};
      \node[draw,circle] (8) at (-2,0.3) {8};
      \node[draw,circle] (9) at (-2,-1) {9};
      \node[draw,circle] (10) at (-3,-2) {10};
      \node[above=0.1em of 1, blue] {$r_1 = 7$};
      \node[right=0.1em of 5, blue] {$r_5 = 10$};
      \node[above=0.1em of 6, blue] {$r_6 = 8$};
      \node[below=0.1em of 7, blue] {$r_7 = 4$};
      \node[below=0.1em of 2, blue] {$r_2 = 2$};
      \node[below=0.1em of 3, blue] {$r_3 = 3$};
      \node[right=0.1em of 9, blue] {$r_9 = 5$};
      \node[below=0.1em of 10, blue] {$r_{10} = 15$};
      \node[left=0.1em of 8, blue] {$r_{8} = 1$};
      \node[above=0.1em of 4, blue] {$r_{4} = 5$};
    
      % Draw the highlighted edges with labels
      \draw[->,line width=0.6mm, draw=orange] (0) -- (1) node[midway, left] {\scriptsize 1};
      \draw[->,line width=0.6mm, draw=orange] (1) -- (5) node[pos=0.4, right] {\scriptsize 5};
      \draw[->,line width=0.6mm, draw=orange] (5) -- (0) node[midway, below] {\scriptsize 1};
      \draw[->,line width=0.6mm, draw=green] (5) to[bend right=20] (0);
      \draw[->,line width=0.6mm, draw=orange] (6) -- (5) node[pos=0.3, left] {\scriptsize 7};
      \draw[->,line width=0.6mm, draw=orange] (5) -- (6);
      \draw[->,line width=0.6mm, draw=green] (7) -- (5) node[pos=0.5, right] {\scriptsize 2};
      \draw (2) -- (5) node[midway, left] {\scriptsize 2};
      \draw[->,line width=0.6mm, draw=green] (2) -- (7) node[midway, below] {\scriptsize 2};
      \draw (2) -- (0) node[midway, left] {\scriptsize 7};
      \draw[->,line width=0.6mm, draw=green] (0) -- (4) node[pos=0.4, right] {\scriptsize 2};
      \draw (4) -- (1) node[midway, above] {\scriptsize 4};
      \draw[->,line width=0.6mm, draw=green] (4) -- (8) node[midway, above] {\scriptsize 2};
      \draw[->,line width=0.6mm, draw=green] (8) -- (9) node[pos=0.3, left] {\scriptsize 3};
      \draw[->,line width=0.6mm, draw=purple] (0) -- (9) node[midway, above] {\scriptsize 5};
      \draw[->,line width=0.6mm, draw=purple] (9) -- (0);
      \draw[->,line width=0.6mm, draw=green] (9) -- (3) node[pos=0.2, below] {\scriptsize 2};
      \draw[->,line width=0.6mm, draw=green] (3) -- (2) node[midway, below] {\scriptsize 3};
      \draw[->,line width=0.6mm, draw=purple] (10) -- (9) node[pos=0.3, above] {\scriptsize 1};
      \draw[->,line width=0.6mm, draw=purple] (9) -- (10);
      
      % \draw[->, line width=1mm, red] (B) -- (C) node[midway, right] {Label 2};
      % \draw[->, line width=1mm, red] (C) -- (A) node[midway, left] {Label 3};
    \end{tikzpicture}
        }
    \caption{\label{fig:solution} A GTSP-rd solution, containing the three routes $\mathcal{R}_1,\mathcal{R}_2,\mathcal{R}_3$.           
    }

\end{center}
\end{figure}

Given the set of solutions, we focus on those that optimize two objective functions also explored other
 works \citep{archetti2011complexity,archetti2015complexity,reyes2018complexity,montero2023solving}.
   For the first, a deadline $D$ to complete all routes is given, and 
   it seeks to minimize the total distance traveled $\sum_{i=1}^{x} d_{\mathcal{R}_i}$ 
   (\textbf{GTSP-rd(distance)}). This type
    of objective function is also known as total sum. To the second, no deadline is 
   given and the total time needed to complete all routes $T_{\mathcal{R}_{x}}+d_{\mathcal{R}_{x}}$ is minimized
    (\textbf{GTSP-rd(time)}), that is, minimize the makespan.   

When all release dates are equal, that is,  $r_1 = r_2 = \cdots = r_n$, the GTSP-rd (time) and GTSP-rd (distance) problems 
are equivalent.
Furthermore, GTSP and GTSP-rd are also equivalent in this scenario, making GTSP a special case of GTSP-rd. Hence, the GTSP-rd problem is NP-Hard for both 
objective functions.  However, in \citet{archetti2015complexity} and \citet{reyes2018complexity} was demonstrated 
that for certain graph classes polynomial solutions exist.

We aim to explore the GTSP-rd within special graph classes, discerning the levels of complexity and identifying potential gaps 
for efficient solutions. This examination delineates the boundaries of tractability and the challenges posed by
 various graph structures. In this paper, we deal with a fundamental graph class, the paths. The results are detailed in
  followings sections.

 \section{GTSP-rd on paths}\label{sec:on_path}  

 A \textit{path} is a simple graph  $P = (V,E)$ whose vertices can be arranged in a linear sequence 
 in such a way that two vertices are adjacent if they are consecutive in the sequence \citep{bondy2008graph}, that is,
  $V = \{v_0,v_1,\dots,v_k\}$, $E = \{(v_0,v_1), (v_1,v_2), \dots , (v_{k-1},v_k) \}$.
The vertices $v_0$ and $v_k$ are the \textit{extremities}, and the vertices $v_1,\dots,v_{k-1}$ are \textit{internal vertices} of $P$.
 
To address the GTSP-rd on paths, we first consider the special case where one of the 
extremities of the path is the depot, specifically the path $P = (\{0\} \cup N, E = \{(v_i, v_{i+1}) : 0 \leq i \leq n-1\})$.
 Without loss of generality, we assume that vertex $0$ is the left extremity and that the vertices are ordered as $1, \dots, n$. 
After that, we use this special case to construct a solution to general case where the depot is located anywhere in path.

 The previous works \citep{archetti2015complexity, reyes2018complexity} proposes algorithms to solve 
 the GTSP-rd in paths (Table~\ref{tbl1}).
 They are related works, since the first one solves the general path case 
 in $O(n^3)$ and the second one solves to special case where the depot is 
 at a path extremity in $O(n^2)$. Both of them using dynamic programming as technique.

    \begin{table}[h]
        \caption{Previous complexity results at GTSP-rd for paths.}\label{tbl1}
        
        \begin{center}
        \begin{tabular}{|c|c|c|c|}
            \hline 
        \textbf{Paper}                          & \textbf{Path}      & \textbf{time}     & \textbf{distance} \\ \hline
        \multirow{2}{*}{\cite{archetti2015complexity}} & Depot at extremity & -                 & -                 \\ \cline{2-4} 
                                                & General case       & $O(n^3)$          & $O(n^3)$          \\ \hline 
        \multirow{2}{*}{\cite{reyes2018complexity}}    & Depot at extremity & $O(n^2)$          & $O(n^2)$          \\ \cline{2-4} 
                                                & General case       & -                 & -                 \\ \hline
        \end{tabular}
        \end{center}
        
        \end{table}

In the following sections, we explore proposed solutions for the GTSP-rd on paths found in the literature,
while provide some enhancements to these approaches. In the Section \ref{subsec:path_depot_end} we explore the
     special case where the depot is an extremity and in Section \ref{subsec:path_depot_anywhere} we explore the
      general case.

\section{Special case: depot on an extremity}\label{subsec:path_depot_end}

To simplify, when discussing a path with the depot at an extremity, we denote the distance from vertex $u$ to
vertex $0$ (depot) as $\tau_u$, this distance can be found using a traverse graph algorithm (such as Depth First Search) 
along the path, where for a path $P' \subseteq P$ with vertex $u$ at one extremity and vertex $0$ at the other,
 $\tau_u = \sum_{e \in E(P')} d_e$.
   
The subsequent properties and definitions are crucial for solving the problem. Although they are addressed in
\citet{reyes2018complexity} and \citet{archetti2015complexity}, here we adjust them according to the adopted notation.
   
First we assume that $r_i \leq r_{i+1}$ for $i \in N$. With this assumption, Proposition~\ref{prop:p3} can also 
       be applied.
    
    \begin{proposition}\textbf{\cite{archetti2015complexity}}\label{prop:p3}
        Given two vertices $i$ and $j$ such that $i<j$, if $\tau_i < \tau_j$ then there is exist an optimal solution such 
        that $i$ and $j$ are delivered in the same route. 
    \end{proposition}
    
Thus, from this point forward, we only need to consider instances with pairs $i,j$ where
    $i<j \text{ and } \tau_i \geq \tau_j$ as input. 
   When it is false, we can simply disregard $i$ and retain only $j$ in the path
     resulting in the same solution.
    
    \begin{definition}\textbf{\cite{reyes2018complexity}}
       Two routes $\mathcal{R}_1$ and $\mathcal{R}_2$ with $min\{r_i \text{ } | \text{ }  i \in S_{1}^{d}\} < min\{r_j \text{ } | \text{ } j
       \in S_{2}^{d}\}$ are \textit{non-interlacing} if and only if $max\{r_i \text{ } | \text{ } i \in S_{1}^{d}\} < min\{r_j \text{ } 
       | \text{ } j\in S_{2}^{d}\}$.
    \end{definition}

   \begin{definition}\textbf{\cite{reyes2018complexity}}
       A solution $\mathcal{X} $ containing non-interlacing routes can be characterized by the set of
       customers with highest index in each route, i.e., $\mathcal{X} = \{v_1, v_2, \dots , v_k, n\}$ with $1 \leq v_1 \leq v_2 \leq \dots \leq v_k \leq n$, 
       indicating that customers $S_{1}^{d} = \{1, \dots , v_1\}$ are attended on the first route, orders $ S_{2}^{d} = \{v_{1 + 1}, \dots , v_2\}$ 
       are delivered on the second route, and so on.
    \end{definition}

    \begin{lemma}\textbf{\cite{reyes2018complexity}}\label{prop:p1}
        Any feasible solution for a GTSP-rd on path with the depot on an extremity can be transformed into a feasible 
        solution with non-interlacing routes, without increase in the total travel time.
    \end{lemma}
   
    Following the Lemma~\ref{prop:p1},we can construct a solution to GTSP-rd (time) for the special case of paths where
     the depot is located at one extremity. This solution uses only non-interlacing routes, meaning routes formed by
      contiguous sequences of vertex indices.

    \begin{lemma}\textbf{\cite{reyes2018complexity}}\label{prop:p2}
        If a solution minimizing the completion time exists, then there exists an optimal
        non-interlacing solution $\mathcal{X} = \{v_1, v_2, \dots , v_k, n\}$ with the property that each partial solution
        $\mathcal{X}_{[1,p]} = \{v_1, v_2, \dots , v_p\}$, delivering orders $\{1,\dots, v_p\}$, for $p = 1, . . ., k$, completes
         the delivery of these order subsets as early as possible.
    \end{lemma} 
   
    The Lemma~\ref{prop:p2} shows that this problem has an optimal substructure property, it enables the dynamic 
    programming approaches described below.

   In the following sections, we present dynamic programming approach proposed by \citet{reyes2018complexity},
    and provide optimizations to reduce the time complexity for both GTSP-rd (time) (\ref{sec:improve_n}) and GTSP-rd (distance)
     (\ref{sec:improve_nlog}).
   
   \subsection{GTSP-rd (time)}\label{sec:tsp_rd_time}
   
    At the recurrence proposed (Equation~\ref{eq:reyes}), $c(i)$ calculates the minimum completion time 
    to attend the customers $\{1,\dots,i\}$:
    
    \begin{equation}\label{eq:reyes}
        c(i) = 
        \begin{cases}
             0 , \quad \text{if } i = 0 \\
             \min_{ 0 \leq j < i} \{\max\{c(j),r_i\} + 2\tau_{j+1}\}, \quad \text{otherwise.}
        \end{cases}
    \end{equation}
    
    Consider the recursive step to compute $c(i)$. The delivery to customer $i$ can be included in two types of routes: the 
    first one, is along with other delivery customers (vertices) $j+1, \dots, i-1$, and it is added to the partial solution 
    attending the customers $\{1, \dots, j\}$ when $j \leq i-2$. The second type of route involves creating a new route 
    containing only $i$ when $j = i-1$. In both cases, the minimum completion time for the new route including $i$ is the 
    earliest possible dispatch time for this route, $\max\{c(j), r_i\}$, added to the travel time of route that would 
    be $2\tau_{j+1}$.
    
   It is not difficult to observe that this recurrence relation can be calculated in \( O(n^2) \), since for each \( i \in N \) we need
    to compute the terms \( \max\{c(j),r_i\} + 2\tau_{j+1} \) for minimization where \( 0 \leq j < i \). 
    However,  we will demonstrate how to compute this equation in \( O(n) \). For this, we need to show that the 
    function \( c(i) \) is non-decreasing.
    
    \begin{lemma}\label{lema:cresc}
        The function $c(i)$ is non-decreasing.
    \end{lemma}
    
    \begin{proof}
        By contradiction, let's suppose that exist a $k \in N$ that $c(k) > c(k+1)$. In general case, this is equivalent to:
   
        \vspace{-1em}
        \begin{equation*}
            \min_{ 0 \leq j < k} \{\max\{c(j),r_k\} + 2\tau_{j+1}\} > \min_{ 0 \leq j < k+1} \{\max\{c(j),r_{k+1}\} + 2\tau_{j+1}\}
        \end{equation*}
    
        We can rewrite the right side of the inequality in the way we explicit the expression when $j = k$.
        
        \vspace{-1em}
        \begin{align*}
            \min_{ 0 \leq j < k} \{\max\{c(j),r_k\} + 2\tau_{j+1}\} > \min\{  &
                \min_{ 0 \leq j < k}\{\max\{c(j),r_{k+1}\} + 2\tau_{j+1}\}, \\ 
                & \max\{c(k),r_{k+1}\} + 2\tau_{k+1}
            \}
        \end{align*}
    
        If this is true, then the two inequalities bellow are also true:

        \begin{align*}
            \min_{ 0 \leq j < k} \{\max\{c(j),r_k\} + 2\tau_{j+1}\} > \min_{ 0 \leq j < k}\{\max\{c(j),r_{k+1}\} + 2\tau_{j+1}\} \quad \quad (I)
        \end{align*}
   
        \vspace{-2em}
   
        \begin{align*}
            \min_{ 0 \leq j < k} \{\max\{c(j),r_k\} + 2\tau_{j+1}\} > \max\{c(k),r_{k+1}\} + 2\tau_{k+1} \quad \quad (II)
        \end{align*}
    
        For the Inequality I: (i) when $c(j) \geq r_{k+1}$ (consequently $c(j) \geq r_k$) the inequality is false because two sides are equal;
        (ii) when $c(j) < r_k$ (consequently $c(j) < r_{k+1}$), the inequality is false because we are considering that $r_k \leq r_{k+1}$, and the right 
        side is greater or equal to left side; (iii) when $c(j) \geq r_k$ and $c(j) < r_{k+1}$, the inequality is false 
        because $r_k \leq c(j) < r_{k+1}$, and the right side is greater; (iv)
        The case when $c(j) < r_k$ and $c(j) \geq r_{k+1}$ is impossible, because $c(j) < r_k \leq r_{k+1} \leq c(j)$ is a
         contradiction.
        
        \vspace{1em}

        For the Inequality II, by definition of $c(k)$, we can rewrite as:
       
         \vspace{-1em}
   
        \begin{align*}
            c(k) > \max\{c(k),r_{k+1}\} + 2\tau_{k+1}
        \end{align*}
    
        It's trivial to see that this inequality is false. When $c(k) > r_{k+1}$ or $c(k) < r_{k+1}$ the right side has the greater value of inequality.  
        
        Given that both inequalities are false, it follows that the inequality $c(k) > c(k+1)$ is also false for any $k \in N$.
         Consequently, since such a $k$ does not exist, it follows that the function $c(i)$ is non-decreasing.   
    \end{proof}
    
    \subsubsection{Proposed Solution for GTSP-rd (time) in \(O(n)\)}\label{sec:improve_n}
        
    As $c(i)$ is a non-decreasing function (Lemma~\ref{lema:cresc}), we can divide the general case of
     Equation~\ref{eq:reyes} into two parts. For a fixed $i$, the first term of the 
    sum $\max\{c(j),r_i\} + 2\tau_{j+1}$ will be $r_i$ as long as the inequality $c(j) \leq r_i$ holds true. 
    Otherwise, $c(j)$ will be the first term of the sum.
    Let's define $k$ as the last value of $j$ such that  $c(j) \leq r_i$. Formally,
     $k = max_{0\leq j \leq n}\{j \text{ }|\text{ } c(j) \leq r_i\}$. So, the base case remain the same $c(0)=0$, 
     but the general case of the Equation~\ref{eq:reyes} can be written in the following way:

    \begin{equation}\label{eq:f_div}
        c(i) = \min_{ 0 \leq j < i}
        \begin{cases}
        r_i + 2\tau_{j+1}, \quad \text{if } j \leq k\\
        c(j) + 2\tau_{j+1}, \quad \text{otherwise }
        \end{cases}
    \end{equation}
    
    It means that we have two sets of routes we can add customer $i$. If $j \leq k$, then the previous constructed route has already finished, and 
    the earliest possible dispatch time is $r_i$. Otherwise, the previous route will finish after the $i$ release, and the earliest possible 
    dispatch time is $c(j)$.
    
    The Equation~\ref{eq:f_div} can be rewritten as follows: 
    
    \begin{equation} \label{eq:e1}
        c(i) = \min \{
            \min_{ 0 \leq j \leq k}\{ r_i + 2\tau_{j+1}\},
            \min_{ k < j < i}\{ c(j) + 2\tau_{j+1}\}
        \}
    \end{equation}
   
    Given Equation~\ref{eq:e1}, we will demonstrate that $c(n)$ can be computed in $O(n)$.
    To achieve this, we must establish that $k$ (Lemma~\ref{lema:lk}), 
     $\min_{0 \leq j \leq k}\{ r_i + 2\tau_{j+1}\}$ (Lemma~\ref{lema:l2}), and $\min_{k < j < i}\{ c(j) + 2\tau_{j+1}\}$ (Lemma~\ref{lema:l3}) can be
      determined in constant time ($O(1)$).
    
    \begin{lemma}\label{lema:lk}
        $k = max_{0\leq j \leq n}\{j \text{ } | \text{ } c(j) \leq r_i\}$ is calculated in $O(1)$ for some $i \in [1\dots n]$.
    \end{lemma}
    
    \begin{proof}

       Let's demonstrate the process of computing $k$ for each $i \in [1 \dots n]$. We start by initializing $k=0$ for the first iteration when $i=0$. For all subsequent iterations, we denote $k'$ as 
        the value of $k$ from the previous step. The process involves iterates over $j$, beginning at $k'$, while $c(j) \leq r_i$.
        The iteration stops when this condition is no longer true, and the new value of $k$ is set to the last $j$ such that $c(j) \leq r_i$. Since $c(j)$
         and the release dates $r_i$ are sorted, there is no need to consider values of $j \leq k'$. This follows   
        from the fact $c(j) \leq c(k')$, and $c(k') \leq r_{i-1} \leq r_i$. The search for new $k$ start at $k'$, which avoids
        unnecessary computations. Over the entire process, the variable $j$ ranged sequentially from $0$ to $n-1$ in the worst case. Hence, the total time
        complexity for computing $k$ for all $i \in [1 \dots n]$ is $O(n)$. Furthermore, each update of $k$ for a specific $i$ is
        performed in constant time $O(1)$.

        % Let $k'$ denote the value of $k$ from the previous iteration (when $i$ was $i-1$). 
        % Additionally, for the initial iteration where $i=0$, 
        % we set $k=0$. So, for each 
        % $i \in [1 \dots n]$, we iterate $j$ starting from $k'$ until $c(j) \leq r_i$ is no longer true. So, the new value of $k$ is equal to the
        %  last value of $j$ where the inequality is true.
        % As the function $c$ and the values of release dates are sorted, there's no need to search for $c(j)$ such that $j < k'$, 
        % it follows from the fact that $c(j) \leq c(k')$ and $c(k') \leq r_{i-1} \leq r_{i}$. Hence, 
        % we can begin the search for a new $k$ from $k'$.
        % Upon completing these operations for all $i \in [1 \dots n]$, the variable $j$ will have ranged from $0$ to $n-1$ in the worst-case. 
        % Overall, the process requires $O(n)$ time to execute entirely, with each choice of $k$ for $i \in [1 \dots n]$ accomplished in constant 
        % time, $O(1)$.

    \end{proof}

   Once we have calculated the $k$, the Lemma~\ref{lema:l2} shows that $\min_{ 0 \leq j \leq k}\{ r_i + 2\tau_{j+1}\}$ can be transformed into a constant sum.
   
    \begin{lemma}\label{lema:l2}
        $\min_{ 0 \leq j \leq k}\{ r_i + 2\tau_{j+1}\}$ is calculated in $O(1)$ for some $i \in [1 \dots n]$.
    \end{lemma}
    
    \begin{proof}
       As showed in Proposition~\ref{prop:p3}, we work only with instances were for some pair $i,j$, $i<j \Rightarrow  \tau_i \geq \tau_j$. It ensures 
       that the distance array $\tau$ is ordered in non-increasing way. Hence, $\min_{ 0 \leq j \leq k}\{ r_i + 2\tau_{j+1}\} = r_i + 2\tau_{k+1}$.
        It is valid because $r_i$ is constant for each $i \in [1 \dots n]$ and $\tau_{k+1} \leq \tau_j$ for each $j \in [1 \dots k]$. Therefore,
         $\min_{ 0 \leq j \leq k}\{ r_i + 2\tau_{j+1}\}$
        can be substituted by $r_i + 2\tau_{k+1}$ and calculated in $O(1)$ time for some $i \in [1 \dots n]$ and in $O(n)$ to calculate 
        for all $i \in [1 \dots n]$.
        
    \end{proof}
    
    As demonstrated in Lemma~\ref{lema:l2}, $\min_{0 \leq j \leq k}\{ r_i + 2\tau_{j+1}\}$ simplifies to $r_i + 2\tau_{k+1}$. This occurs because all 
    customers $j$ where $0 \leq j \leq k$ have completion times $c(j)$ smaller than $r_i$, allowing us to insert customer $i$ into a route with customers
     $j+1, \dots, i-1$. Among these routes, the best route to minimize completion time is the one closest to the depot. Consequently,
      we choose the route with customers $k+1, \dots, i-1$ to include customer $i$ in the same route. We can rewrite recurrence of Equation~\ref{eq:e1} as:
   
   \begin{equation}\label{eq:e2}
     c(i) = \min \{
         r_i + 2\tau_{k+1},
         \min_{ k < j < i}\{ c(j) + 2\tau_{j+1}\}
     \}
   \end{equation}
    
   Equation~\ref{eq:e1} is now expressed as a minimization of a sum, combined with a larger minimization over 
   the interval $[k+1,i-1]$. The Lemma~\ref{lema:l3} show how to compute this efficiently.

    \begin{lemma} \label{lema:l3}
        $\min_{ k < j < i}\{ c(j) + 2\tau_{j+1}\}$ is calculated in $O(1)$ for some $i \in [1 \dots n]$.
    \end{lemma}
    
    \begin{proof}
    
        We show that $\min_{ k < j < i}\{ c(j) + 2\tau_{j+1}\}$ can be calculated in $O(n)$ to all $i \in [1 \dots n]$ and consequently in 
        a constant time for some $i \in [1 \dots n]$. 
    
        Firstly, we define $a_j = c(j) + 2\tau_{j+1}$, and the objective is to find the minimum value of $a_j$ such 
        that $k < j <i$. To achieve this, we can use a structure called \texttt{minqueue}, which is nothing but 
        a queue with a \texttt{find\_min} operation, besides the operations of \texttt{enqueue(x)} and \texttt{dequeue(x)}. These operations
        could be implemented in amortized $O(1)$ time complexity \citep{brass2008advanced}, or $O(1)$ in the worst case \citep{sundar1989worst}. 
    
        As demonstrated in Lemma~\ref{lema:lk}, for each $i$, we efficiently generate a corresponding $k$ in $O(1)$ time,
         where $k < i$. With each iteration of $i$, we insert the element $a_{i-1}$ into the \texttt{minqueue} (\texttt{enqueue($a_{i-1}$)}) 
         and subsequently remove all $a_j$ (\texttt{dequeue($a_{j}$)}) for $j \in [k'+1, k]$, where $k'$ represents the value of $k$ from the 
         previous iteration.
    
        Thus, we show that $O(n)$ insertion operations and $O(n)$ removal operations will be performed. We can thus conclude 
        that computing $\min_{ k < j < i}\{ c(j) + 2\tau_{j+1}\}$ for each $i \in [1,n]$ is performed in constant time $O(1)$.
        
    \end{proof}
    
    \begin{theorem}
       The time complexity to calculate \(c(n)\) using Equation~\ref{eq:e2} is \(O(n)\).
     \end{theorem}

   \begin{proof}
   Lemmas \ref{lema:lk}, \ref{lema:l2}, and \ref{lema:l3} demonstrate that the two terms in the recurrence relation of 
   Equation~\ref{eq:e2} can be computed in constant time for each $i \in [1 \dots n]$. As the operation to determine the
    minimum of two values has constant cost, \(c(n)\) can be calculated in \(O(n)\).
   \end{proof}

    The Algorithm~\ref{alg} calculates $c(n)$ with following the Lemmas \ref{lema:lk}, \ref{lema:l2}, and \ref{lema:l3} 
    and its operations.
    %$i \in [1,n]$ 
    \begin{algorithm}[h]
        \DontPrintSemicolon
        \SetKwInOut{Input}{Input}
        \SetKwInOut{Output}{Output}

        \Input{$n$: number of vertices, $r_i$: release dates, $\tau_i$: travel time from $i$ to depot, $i \in [1, n]$.}
        \Output{$c[i]$: completion time for each node $i \in [1, n]$.}

            $Q \gets minqueue()$ \;
            $k \gets 0$ \;
           $c[0] \gets 0$ \;       
            $c[i] \gets \infty \quad \forall i \in [1,n] $ \;
   
            \;
            \For{$i \leftarrow 1$ \KwTo $n$}{
   
                \While{$k < n$ \textbf{and} $c(k+1) \leq r_i$}{
                    $k \gets k + 1$ \;
                    Q.\texttt{dequeue($a_{k}$)} \;
                }   
   
               \;

                $\min_{ki} \gets $  $Q$.\texttt{find\_min()} \;
                $c[i] \gets min(r_i + 2\tau_{k+1},\min_{ki} )$\;
                $a_i \gets c[i] + 2\tau_{i+1}$ \;
   
                \;
                Q.\texttt{enqueue($a_{i-1}$)}\; 
            }
    
        \caption{GTSP-rd(time) MinQueue in path with depot at the extremity}\label{alg}
    \end{algorithm}

    % \setAlgoLined
    % \Input{$3$ discos de tamanhos diferentes; e 3 hastes ($0$ (origem), $1$ (auxiliar), $2$ (destino)) de capacidade diferente( Haste 0: 3 discos, Haste 1: 2 discos e Haste 2: 1 disco). Estado inicial: discos distribuídos de forma aleatória.}
    % \Saida{Estado final: todos os discos empilhados na Haste $0$ (origem).}

   \subsection{GTSP-rd (distance)}\label{sec:tsp_rd_dist}
   
   In this section, we discuss a scenario where the objective is to minimize the total distance traveled by the Traveling Salesman. If we 
   consider a version of this problem where there's no final deadline $D$ (or if $D$ is sufficiently large), the optimal strategy is to wait until all
    packages are available before initiating deliveries, thus completing a single comprehensive route that includes all customers. The final cost will 
    be $2\tau_1$.
   
    A similar approach could be employed when there is a final deadline $D$. We wait as long as possible to initiate deliveries, incorporating all
     customers who are already available at the time of the latest dispatch. Then, we begin the second route with the first 
     customer $u$ who hasn't been included in the first route by the time of the latest dispatch for $u$, and continue in this 
     manner until there are no more customers left.
   
    Given that $\lambda(i)$ represents the latest time to dispatch customer $i$ in order to serve customers $\{i \cdots n\}$
     with non-interlacing routes, 
    the minimum total distance traveled by these routes is $D - \lambda(1)$. It happens because the latest time to dispatch 
    customer $1$ depends on the latest dispatch time of the next routes, so we do not have waiting time between
    two routes. This idea led \citet{reyes2018complexity} to formulate the 
    following recurrence relation:
    
   \begin{equation}\label{eq:reyes2}
       \lambda(i) = 
       \begin{cases}
            D , \quad \text{if } i = n+1 \\
            \max_{ j > i} \{ \lambda(j) - 2\tau_i  \mid  \lambda(j) - 2\tau_i \geq r_{j-1} \}, \quad \text{otherwise}
       \end{cases}
   \end{equation}
   
   The base case involves introducing a hypothetical customer $n+1$, representing the final deadline $D$.  
   As we are working with non-interlacing routes to determine  the latest time to dispatch customer $i$, 
   we attempt to incorporate it into all previously established routes, such as $\{i, \cdots , j-1\}$.
    To postpone 
    dispatch as much as possible, the latest time to dispatch this route will be $\lambda(j)$ (the latest time for the next route dispatch) minus the cost of 
    executing this route, which is $2\tau_i$.
   
   The feasibility condition $\lambda(j) - 2\tau_i \geq r_{j-1}$ ensures that the time of dispatch from this 
   route is at least the greatest release date of the customers into this route, which is $r_{j-1}$.
   
   Therefore, it is not difficult to observe that this recurrence relation can be calculated in \( O(n^2) \) time complexity,
    since for each $i \in N$, we need to compute the terms $\lambda(j) - 2\tau_i$ and choose the maximum among them, for 
    each $j \in [i+1,n+1]$.

   \subsubsection{Proposed Solution for GTSP-rd (distance) in $O(n \log \log n)$}\label{sec:improve_nlog}
   
   In this section we describe how to modify the Equation~\ref{eq:reyes2} to calculate it in $O(n \log \log n)$.
   
   We can rewrite the feasibility inequality $\lambda(j) - 2\tau_i \geq r_{j-1}$ as $\lambda(j) - r_{j-1} \geq 2\tau_i$. 
   Hence, the general case of Equation \ref{eq:reyes2} can be rewrite as:
   
   \begin{equation}\label{eq:dist}
       \lambda(i) = \max_{j > i} \{ \lambda(j) \mid \lambda(j) - r_{j-1} \geq 2\tau_i \} - 2\tau_i
   \end{equation}

   The base case remains the same, $\lambda(n+1) = D$. The Equation~\ref{eq:dist} can be calculated in $O(n \log \log n)$ 
   using auxiliary heaps \citep{thorup2000ram,williams1964algorithm}. The Lemma \ref{theom:dist_1} show how to calculate 
    the Equation~\ref{eq:dist} in $O(n \log \log n)$.
   
   \begin{theorem}\label{theom:dist_1}
       The complexity to calculate $\lambda(1)$ through Equation~\ref{eq:dist} is  $O(n \log \log n)$.
   \end{theorem}
   
   \begin{proof}
       Here, we will utilize the heap structure introduced by \citet{thorup2000ram}, which supports the operations 
       \texttt{insert}, \texttt{remove}, and \texttt{find\_min} (\texttt{find\_max}) with the following computational costs: 
       $O(\log \log n)$ for \texttt{insert}, amortized $O(\log \log n)$ for \texttt{remove}, and $O(1)$ for \texttt{find\_min} (\texttt{find\_max}).
       
       If we didn't have the feasibility inequality, it would suffice to find the largest value of $\lambda(j)$ for all $j > i$. However, some $j$'s
       don't respect the inequality and should not be considered in maximization. 
       
       To calculate $\lambda_{max} = \max_{j > i} \{ \lambda(j) \mid \lambda(j) - r_{j-1} \geq 2\tau_i \}$ we maintain 
       a max-heap $H_1$ containing only the values of $\lambda(j)$ that the inequality are true. 
       
       As the computation of $\lambda(i)$ depends on all $j > i$, then we compute it from $n$ to $1$. Then, 
       $\forall i \in \{n, n-1, \ldots, 1\}$ we use the operation \texttt{find\_max} in $H_1$,
       $\lambda(i) = \lambda_{max} - 2\tau_i$. After it, we use the operation  \texttt{insert($\lambda(i)$)} into
        $H_1$ to be used in the next iterations of $i$. 
       
       We must guarantee that for the actual iteration $i$, $H_1$ has only elements $\lambda(j)$ that $\lambda(j) - r_{j-1} \geq 2\tau_i$.
       Given that $h_1$ are the set of elements in $H_1$ in the current iteration $i$ and $h_1'$ are the elements
       in $H_1$ in the next iteration $i-1$, then $h_1' \setminus \{\lambda(i)\} \subseteq h_1$. 
       
       It will follow from the fact that between two iterations, two operations must be done. First, 
       insert $\lambda(i)$. Proposition~\ref{prop:p3} establishes that $2\tau_{i} \leq 2\tau_{i-1}$, and no new 
       element will be inserted. On the second operation, the elements where $ 2\tau_{i-1} > \lambda(j) - r_{j-1} \geq 2\tau_{i}$ 
       will be removed from $H_1$.
       
       % and the following 
       % inequalities will hold: $\lambda(j) - r_{j-1} \geq 2\tau_i$ and $\lambda(j) - r_{j-1} \geq 2\tau_{i-1}$. 
       
       A quick and efficient way to perform this removal is by using an auxiliary min-heap $H_2$. It will maintain the
        items $a_j = \lambda(j) - r_{j-1}$. Each element $a_i$ will be inserted in $H_2$ together when $\lambda(i)$ 
        is inserted in $H_1$. To know the items $\lambda(k)$ that will be removed 
        from $H_1$ we get the minimum $a_{k} \in H_2$ and 
        remove if $2\tau_{i-1} > a_{k}$ and also remove $a_k$ from $H_2$. It will continue until no more elements violate 
        the inequality.
        
        As max-heap $H_1$ and min-heap $H_2$ takes $n$ operations of insertion and in the worst case $n-1$ operations of
         removal. The complexity to calculate $\lambda(1)$ is $O(n \log \log n)$. 
    \end{proof}

     A drawback from use Thorup's heap is that space cost, which is $O(n2^{\epsilon \omega})$, where $\epsilon$ is 
    any positive constant and $\omega$ is the number of bits used to represent the greater number in heap. Also, there is 
    a randomized implementation giving $O(\log \log n)$ expected time and $O(n)$ space.

    Algorithm~\ref{alg:dist_1} implements Equation~\ref{eq:dist} using the structures described in Theorem~\ref{theom:dist_1}.  
    Its complexity depends on the choice of the auxiliary heap \citep{brodal2013survey}. If binary heaps \citep{williams1964algorithm} are used, 
    the time complexity becomes \( O(n \log n) \).

    \begin{algorithm}[h]
        \DontPrintSemicolon

        \SetKwInOut{Input}{Input}
        \SetKwInOut{Output}{Output}

        \Input{$n$: number of vertices, $r_i$: release dates, $\tau_i$: travel time from $i$ to depot, $i \in [1, n]$.}
        \Output{$\lambda[i]$:  latest time to dispatch customer $i$ in order to serve customers $\{i \cdots n\}$
        with non-interlacing routes.}

            $H_1 \gets maxHeap()$ \;
            $H_2 \gets minHeap()$ \;
   
            \;
   
            $\lambda[n+1] \gets D$ \;
            $H_1$.\texttt{insert}($\lambda[n+1]$) \;
            $H_2$.\texttt{insert}($\lambda[n+1] - r_{n}$) \;   

            \;

            \For{$i \gets n$ \textbf{downto} $1$}{
    
                \While{$H_2$.\texttt{find\_min}() < $2\tau_i$}{
   
                   $ k \gets H_2$.\texttt{minIndex}() \;
   
                   $H_1$.\texttt{remove}($\lambda[k]$)\;
                   $H_2$.\texttt{remove}($H_2$.\texttt{find\_min}())\;
                }   
                \;
   
                $\lambda[i] \gets H_1.\texttt{find\_max}() - 2\tau_i$\;
   
                $H_1$.\texttt{insert}($\lambda[i]$) \;
                $H_2$.\texttt{insert}($\lambda[i] - r_{i-1}$) \;     
            }
    
        \caption{GTSP-rd(distance) Heap in path with depot at the extremity}\label{alg:dist_1}
    \end{algorithm}

    \section{General case}\label{subsec:path_depot_anywhere}
 
    In the general case where the depot can be located anywhere along the path, not just at the extremities, 
    if we remove the depot from path it divides the original path into two disconnected paths.
    Without loss of generality, we denote the set of vertices for these two paths as the \textit{left vertices} $N_l$ 
    and the \textit{right vertices} $N_r$, such that $N = N_l \cup N_r$. Additionally, let $n_l = |N_l|$, $n_r = |N_r|$, and $n_r + n_l = n$.
    
    It's easy to see that a route in a path $P$ with customers belonging to both sets $N_r$ and $N_l$ can be transformed
     into two disjointed routes. Each of these routes exclusively contains customers from one side, maintaining equivalent
      costs. Consequently, we can formulate an optimal solution consisting exclusively of routes 
     comprising customers from the same side.
    
    We use the following notation in this section. To the vertices in $N_l$, release dates and distance are 
    denoted by $r_i^{l}$ and $\tau_{i}^{l}$ respectively. The same are valid to vertices in $N_r$, which are 
    denoted as $r_i^{r}$ and $\tau_{i}^{r}$. As we are only treating two sides separately to build the routes,
     the Proposition \ref{prop:p3} and Lemmas \ref{prop:p1} and \ref{prop:p2} still holds for
     each side separately. Without loss of generality and we relabel the indices of vertices of
      $N_l$ and $N_r$ to $\{1, 2, \ldots, n_l\}$  and $\{1, 2, \ldots, n_r\}$ in such way that $r_i^{l} \leq r_{i+1}^{l}$ 
      for $i \in N_l$ and $r_i^{r} \leq r_{i+1}^{r}$ for $i \in N_r$. 
     Based on Proposition \ref{prop:p3}, to the left (right) side, $i<j \Rightarrow \tau_{i}^{l} \geq \tau_{i}^{l}$ ($\tau_{i}^{r} \geq \tau_{i}^{r}$). 
     Then we have an instance like Figure \ref{fig:p_instance}.
   
   \begin{figure}[h]
       \begin{center}
          \begin{tikzpicture}
               \draw 
               (-6, 0) node[circle, draw] (l1){1}
               (-4.5, 0) node[circle, draw] (l2){2}
               (-3, 0) node (ldot){$\cdots$}
               (-1.5, 0) node[circle, draw](ln){$n_l$}
               (0, 0) node[draw](0){0}
               (1.5, 0) node[circle, draw](rn){$n_r$}
               (3, 0) node (dot){$\cdots$}
               (4.5, 0) node[circle, draw] (2){2}
               (6, 0) node[circle, draw](1){1};
           
               \path[-] (0) edge [bend left=0] node[above, font=\tiny] {} (rn);
               \path[-] (0) edge [bend left=0] node[above, font=\tiny] {} (ln);
               \path[-] (rn) edge [bend left=0] node[above, font=\tiny] {} (dot);
               \path[-] (dot) edge [bend left=0] node[above, font=\tiny] {} (2);
               \path[-] (1) edge [bend left=0] node[above, font=\tiny] {} (2);
               \path[-] (ln) edge [bend left=0] node[above, font=\tiny] {} (ldot);
               \path[-] (ldot) edge [bend left=0] node[above, font=\tiny] {} (l2);
               \path[-] (l1) edge [bend left=0] node[above, font=\tiny] {} (l2);
               
           \end{tikzpicture}
       \end{center}
       \caption{An arbitrary Path instance with the depot in an arbitrary vertex\label{fig:p_instance}.}
   \end{figure}
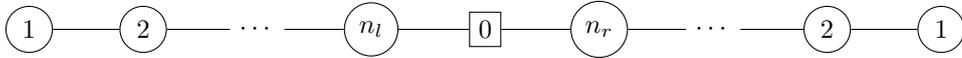   
    
    \subsection{Solution $O(n^2)$ for GTSP-rd (time)}
    
   It can be verified that Lemma~\ref{prop:p1} remains valid for each side of the depot along the general
    case of path. This means we can proceed with our 
   solution construction, employing only non-interlacing routes within the sets $N_l$ and $N_r$ independently.
    Extending the recurrence of Equation~\ref{eq:e2}, 
   we define $c(i,j)$ as the minimum completion time to attend the customers $\{1,\dots,i\} \subseteq N_l$ and $\{1,\dots,j\} \subseteq N_r$. 
   To compute $c(i,j)$, we select the optimal choice between incorporating customer $j$ into a non-interlacing route to the right of depot $R(i,j)$ or 
   including customer $i$ into a non-interlacing route to the left of depot $L(i,j)$.
    
   \begin{equation}\label{eq:rr2}
        c(i,j) = 
        \begin{cases}
            0, \quad \text{if } i = j = 0\\
            \min\{L(i,j),R(i,j)\}, \quad \text{otherwise }
        \end{cases}
    \end{equation}

   As the recurrence relations $L(i,j)$ and $R(i,j)$ for each $i,j$ with $i \in N_l$ and $j \in N_r$ calculate the best choice for each side only containing customers of this side, 
   they are very similar to recurrence relation of Equation \ref{eq:e2}. The function $L(i,j)$ includes a variable analogous
    to $k$, in Equation \ref{eq:e2}.
   The goal is to mark the set of customers that can or can not depart at $r_i^{l}$. 
   But here, for each $i$ we have $n_r$ possible $j$'s then 
   we define $k_j^l = max\{w \in n_l \text{ } | \text{ } c(w,j) \leq  r_i^{l}\}$ for each $j \in n_r$. Similarly, for $R(i,j)$, we define
    $k_i^r = max\{w \in n_r \text{ } | \text{ } c(i,w) \leq  r_j^{r}\}$ for each $i \in n_l$. The functions $L$ and $R$ are defined below:
    
    \begin{equation}\label{eq:rr2_L}
        L(i,j) = \min \{
            r_i^{l} + 2\tau_{k_j^l+1}^{l},
            \min_{k_j^l < w < i}\{ c(w,j) + 2\tau_{w+1}^{l} \}
        \}
    \end{equation}
    
    \vspace{-1em}
   
    \begin{equation}\label{eq:rr2_R}
        R(i,j) = \min \{
            r_j^{r} + 2\tau_{k_i^r+1}^{r},
            \min_{k_i^r < w < j}\{ c(i,w) + 2\tau_{w+1}^{r} \}
        \}
    \end{equation}
    
    On the left, customer $i$ can be integrated into a route along with customers $k_j^l+1,\dots, i-1$ 
    where $k_j^l$ represents the closest customer to on left of depot such that the minimum completion time is less
     than $r_i$. The cost in this case is $r_i^{l} + 2\tau_{k_j^l+1}^{l}$. 
    
    Customer $i$ can also be incorporated into a route alongside customers $w+1,\dots, i-1$ when $w \leq i-2$, or it can form a new route comprising 
    only itself 
    when $w = i-1$. In both cases, the cost is determined by the shortest possible dispatch time for this 
    route, denoted as $c(w,j)$, added to the travel cost of $2\tau_{w+1}^{l}$. From all these possibles ways to
     add $i$ in a route, we chose the one that return the minimum cost.
   Similarly, the same principle applies to customers on the right, including customer $j$.
    
    To solve the recurrence relation using dynamic programming, we evaluate the function $c(i,j)$ for all $n_r \cdot n_l$ 
    possible states, where each computation requires constant time. As a result, the value of $c(n_l, n_r)$ can be determined in $O(n^2)$ time complexity. 
    The following lemmas provide a more detailed explanation of this process.

    \begin{lemma}\label{lemma:subp}
        The number of subproblems of $c(n_l,n_r)$ in Equation~\ref{eq:rr2} are the order of $O(n^2)$.
    \end{lemma}
    
    \begin{proof}
       For any given values of $i$ and $j$ where $i,j \geq 0$, the function $c(i,j)$ is determined by two 
       functions: $L(i,j)$ and $R(i,j)$. The function $L(i,j)$ considers at most the previous $i$ customers 
       from the left when it calculates $c(w,j)$, where $w \in [k_j^l+1,i-1]$. Similarly, the function
        $R(i,j)$ considers at most the preceding $j$ customers from the right when it computes $c(i,w)$, where
         $w \in [k_i^r+1, j-1]$.
   
       To determine $c(n_l,n_r)$, we must compute $c(i,j)$ for all $i \in [0,n_l-1]$ and $j \in [0,n_r-1]$, as well as $c(i,n_r)$ for all $i \in [0,n_l-1]$, 
       and $c(n_l,j)$ for all $j \in [0,n_r-1]$. This entails performing a total of $(n_r \cdot n_l) + n_r + n_l = O(n^2)$ computations in advance, 
       constituting
       the subproblems necessary to solve $c(n_l,n_r)$.
    \end{proof}
   
    \begin{lemma}\label{lemma:constant}
        Given the Equation~\ref{eq:rr2}, for each $i \in N_l$ and $j \in N_r$, $c(i,j)$ can be computed in $O(1)$.
    \end{lemma}
    
    \begin{proof}
       To show that $c(i,j)$ is computed in constant time, we need to show that $L(i,j)$ and $R(i,j)$ are computed in $O(1)$.
       In function $L(i,j)$, we have a minimum calculation between two terms: $r_i^{l} + 2\tau_{k_j^l+1}^{l}$ and  
       $\min_{k_j^l < w < i}\{ c(w,j) + 2\tau_{w+1}^{l} \}$. As in the specific case when the depot is 
       located in an extremity, we need to show that these two terms are computed in $O(1)$. Similarly, to the $R(i.j)$ function.
        Since both terms depends on the variables $k^l$ ($k^r$ for $R$) we also need to show that they can be computed in
         $O(1)$. 
   
       \begin{itemize}
           \item  Given a $j \in N_r$, let's show how to compute $k_j^l$ for each $i \in [1, n_l]$.
           Let $k_j^{l'}$ denote the value of $k_j^l$ from the previous iteration (when $i$ was $i-1$). Additionally, for the initial 
           iteration where $i=0$, we set $k_j^l=0$. So, for each $i \in [1,n_l]$, we iterate $w$ starting from $k_j^{l'}$ until $c(w,j) \leq r_i^l$ is no 
           longer true. So, the new value of $k_j^l$ is equal to the last value of $w$ where the inequality is true.
           Upon completing these operations for all $i \in [1,n_l]$, the variable $w$ will have ranged from $0$ to $n_l-1$ in the worst-case. 
           Overall, for a given $j$, the process requires $O(n_l)$ time to execute entirely, with each choice of $k_j^l$ for $i \in [1,n_l]$ accomplished
            in constant time, $O(1)$.

           % For $i = 1$ we can iterate $w$ one by one from $0$ while $c(w,j) \leq r_1^l$, so we store this value as $k_j^l$. 
           % For each $i \in [2 \dots n_l]$, we can iterate $w$ one by one from the last $k_j^l$ stored while $c(w,j) \leq r_i^l$. After that we have
           % a new value to be stored as $k_j^l$. After the operation execution to all $i \in [1 \dots n_l]$, in the worst case the variable $w$ will have 
           % varied from $0$ to $n_l-1$. This is $O(n_l)$ to execute the entire process and $O(1)$ to execute the choice of $k_j^l$ for
           % each $i \in [1 \dots n_l]$ for a given $j$.
       
           % Given a $i \in N_l$, let's show how to compute $k_i^r$ for each $j \in [1 \dots n_r]$.
           % For $j = 1$ we can iterate $w$ one by one from $0$ while $c(i,w) \leq r_1^r$, so we store this value as $k_i^r$. 
           % For each $j \in [2 \dots n_r]$, we can iterate $w$ one by one from the last $k_i^r$ stored while $c(i,w) \leq r_j^r$. After that we have
           % a new value to be stored as $k_i^r$. After the operation execution to all $j \in [1 \dots n_r]$, in the worst case the variable $w$ will have 
           % varied from $0$ to $n_r-1$. This is $O(n_r)$ to execute the entire process and $O(1)$ to execute the choice of $k_i^r$ for
           % each $j \in [1 \dots n_r]$ for a given $i$.
   
           Given a $i \in N_l$, the $k_i^r$ can be computed similarly for each $j \in [1,n_r]$.
           That is, $O(n_r)$ to execute the entire process and $O(1)$ to execute the choice of $k_i^r$ for each $j \in [1, n_r]$ for a given $i$.
           
          When computing \( c(n_l, n_r) \), we require \( n_r \) variables \( k_j^l \) and \( n_l \) variables \( k_i^r \). The cost
            of computing them is given by $n_r \cdot O(n_l) + n_l \cdot O(n_r) = O(n_l \cdot n_r) + O(n_l \cdot n_r) = O(2 n_l \cdot n_r) = O(n^2)$.
            Thus, the amortized cost per calculation in a single iteration is \( O(1) \).

           \item Given that $k_j^l$ can also be calculated in $O(1)$, $r_i^{l} + 2\tau_{k_j^l+1}^{l}$ is just a sum and can be calculated in $O(1)$ time.
            Analogous to 
           $r_j^{r} + 2\tau_{k_i^r+1}^{r}$.

           \item To compute the term \( \min_{k_j^l < w < i} \{ c(w,j) + 2\tau_{w+1}^{l} \} \) in \( O(1) \), we use  
             the same \texttt{minqueue} from \citet{sundar1989worst}, as discussed in Lemma~\ref{lema:l3}.  
             However, this time we require more than one. Since the minimization operation iterates only over $w$,  
              we can use a separate \texttt{minqueue} for each \( j \in N_l \) to efficiently compute the minimum,  
              as the values of \( c(w,j) \) vary for different \( j \).  
              
           Let's define, for a given $j$, $a_{ij}^l = c(i,j) + 2\tau_{i+1}^l$. Our current objective is to identify the 
           smallest value $a_{wj}^l$ such that $w \in [k_j^l+1,i-1]$. For each $j \in N_r$, we maintain a \texttt{minqueue}
            with a cost of $O(n_l)$ for each queue. This process mirrors the operation described in Lemma \ref{lema:l3}. 
            
            Similarly, to compute the function $R$, for each $i \in N_l$, we maintain a queue with a cost of $O(n_r)$.
           Then, we have $n_r$ queues with the final cost $O(n_l)$ and $n_l$ queues with the final cost $O(n_r)$ which is
            equivalent to $n_r \cdot O(n_l) + n_l \cdot O(n_r) = O(n_l \cdot n_r)+O(n_l\cdot n_r) = 2\cdot O(n_l\cdot n_r) = O(n^2)$. 
            Therefore, in amortized time, each calculation in one iteration requires $O(1)$ time.

       \end{itemize}
    \end{proof}
    
    \begin{theorem}\label{theorem:n_squared}
        The recurrence relation \( c(n_l, n_r) \), given by Equation~\ref{eq:rr2}, can be computed in \( O(n^2) \).
    \end{theorem}
    
    \begin{proof}
        As shown in Lemma~\ref{lemma:subp}, the number of subproblems in computing \( c(n_l, n_r) \) is \( O(n^2) \).  
        Furthermore, by Lemma~\ref{lemma:constant}, each subproblem can be solved in amortized constant time, \( O(1) \).  
        Since both statements hold, we conclude that Equation~\ref{eq:rr2} can be computed via dynamic programming in \( O(n^2) \).
    \end{proof}
    
    The Algorithm~\ref{alg2} presents the dynamic programming approach to solve the Equation~\ref{eq:rr2} utilizing the
     operations delineated in Lemma~\ref{lemma:constant}. The minqueues $QL_j$ for $j \in \{1,\dots,n_r\}$ represent the 
     queues used to compute the equation $L$. The Figure~\ref{example_minques} exemplifies why we need more than one \texttt{minqueue} to caculate $L$.
    An analogous process for $QR_i$ for $i \in \{1,\dots,n_l\}$ and the equation $R$.
   
     \begin{figure}[h]
        \begin{center}
            \includegraphics[width=0.9\textwidth]{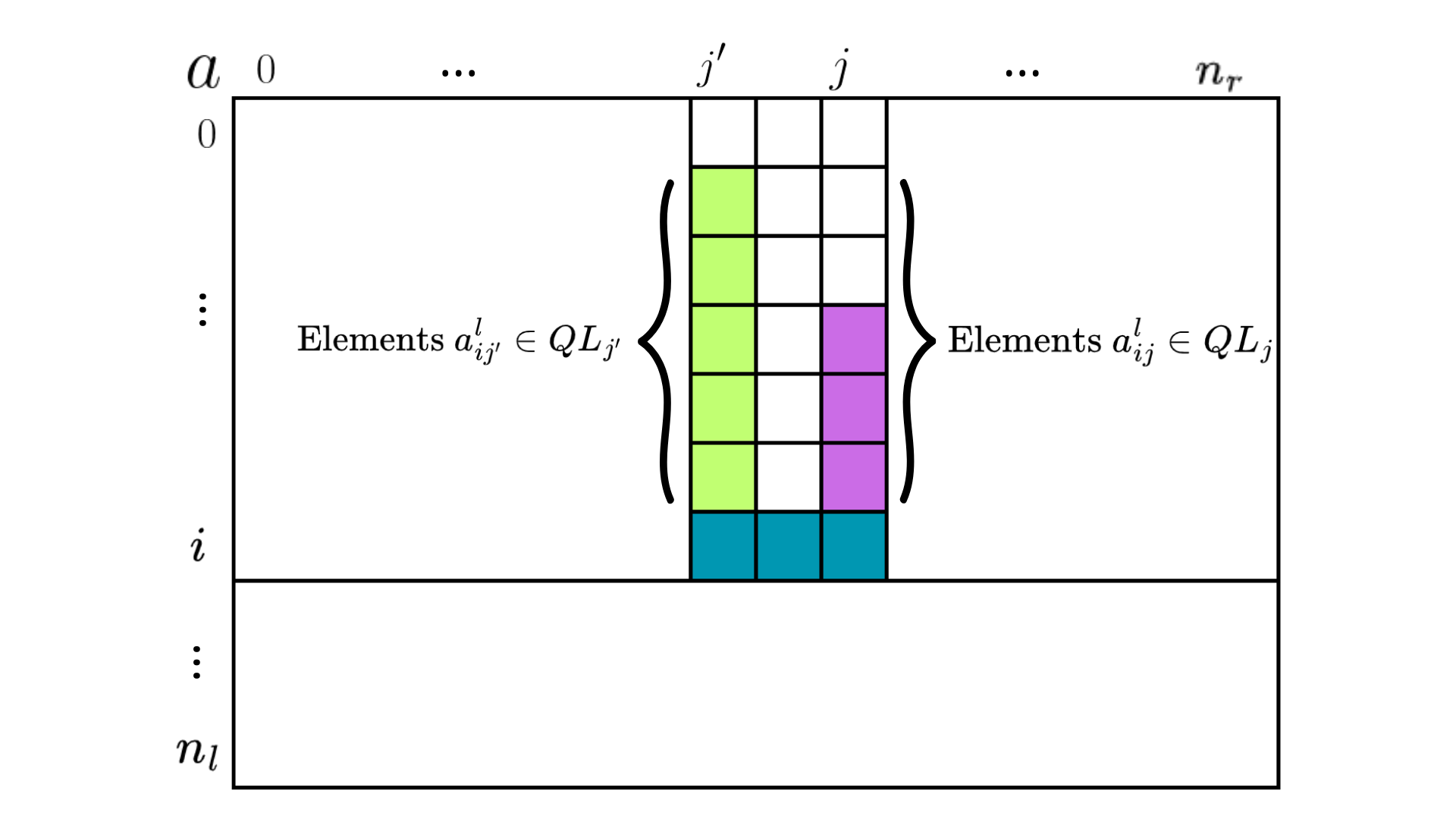}
        
            \caption{\label{example_minques} Example of \texttt{minqueue} usage to calculate $L(i,j)$ to some $i \in N_l$.}           
            
        \end{center}
    
    \end{figure}

    \begin{algorithm}
        \DontPrintSemicolon

        \SetKwInOut{Input}{Input}
        \SetKwInOut{Output}{Output}

        \Input{$n_l$: number of left vertices, $n_l$: number of right vertices, $r$: release dates, $\tau$: travel time from depot.}
        \Output{$c[i,j]$: the minimum completion time to attend the customers $\{1,\dots,i\} \subseteq N_l$ and $\{1,\dots,j\} \subseteq N_r$}

        \;

           $c[i,j] \gets \infty \quad \forall i \in [0,n_l], \forall j \in [0,n_r] $\;
           $k_j^l \gets 0 \quad \forall j \in [0,n_r] $ \;
           $k_i^r \gets 0 \quad \forall i \in [0,n_l] $ \;

           \;

           $QL_j \gets minqueue() \quad \forall j \in [0,n_r]$ \;
           $QR_i \gets minqueue() \quad \forall i \in [0,n_l]$ \;

           \;

           \For{$i \leftarrow 0$ \KwTo $n_l$}{
            \For{$j \leftarrow 0$ \KwTo $n_r$}{
                   
                   \;
                   \If{$i = 0$ \textbf{and} $j = 0$}{
                       $c[i,j] \gets 0$ \;
                   }

                   \;
                   \While{$c[k_j^l,j] \leq r_i^l$}{
                       $QL_j$.\texttt{remove}($a_{k_j^lj}^l$) \;
                       $k_j^l \gets k_j^l + 1$ \;
                   } 
                   \;
   
                   \While{$c[i,k_i^r] \leq r_j^r$}{
                       $QR_i$.\texttt{remove}($a_{ik_i^r}^r$) \;
                       $k_i^r \gets k_i^r + 1$ \;
                   }
                   
                   \;
   
                   $L \gets \min(r_i^{l} + 2\tau_{k_j^l}^{l} , QL_j.\texttt{find\-min()})$\;
                   $R \gets \min(r_j^{r} + 2\tau_{k_i^r}^{r} , QR_i.\texttt{find\-min()})$\;
                   $c[i,j] \gets \min(L,R)$ \;
                   
                   \;
   
                   $a_{ij}^l \gets c[i,j] + 2\tau_{i+1}^l$\;
                   $QL_j$.\texttt{insert}($a_{ij}^l$)\;
                   $a_{ij}^r \gets c[i,j] + 2\tau_{j+1}^r$\;    
                   $QR_i$.\texttt{insert}($a_{ij}^r$)\;
               }
           }
        \caption{GTSP-rd(time) MinQueue in path}\label{alg2}
    \end{algorithm}

   \subsection{Solution $O(n^2 \log \log n)$ for GTSP-rd (distance)}
   
   As in GTSP-rd (time), where the optimal solution consists of routes containing customers from only one side of the depot,  
   we define a recurrence relation similar to Equation~\ref{eq:rr2}.

    Given that $\lambda(i,j)$ represents the latest time to dispatch a route that attend $i \in N_l$ or $j \in N_r$, such that we still have to attend
     the customers
    $\{i \cdots n_l\} \in N_l$ and $\{j\cdots n_r\} \in N_r$. To compute $\lambda(i,j)$, we select the optimal choice between incorporating customer $j$
     into a non-interlacing route to the right of depot $R(i,j)$ or 
    including customer $i$ in a non-interlacing route to the left of depot $L(i,j)$. This is expressed by the following recurrence:
   
   \begin{equation}\label{eq:rr2:dist}
       \lambda(i,j) = 
       \begin{cases}
           D, \quad \text{if } i = n_l, j = n_r\\
           \max\{L(i,j),R(i,j)\}, \quad \text{otherwise }
       \end{cases}
   \end{equation}
   
   The operations of $L$ and $R$ closely resemble Equation~\ref{eq:dist}. In $L$,
    we examine all feasible non-interlacing routes that could involve customer $i$ 
    as the farthest customer from the depot, similarly to $R$.
   
   \begin{equation}\label{eq:rr2_L:dist}
       L(i,j) = \max_{w > i} \{ \lambda(w,j) \mid \lambda(w,j) - r_{w-1}^l \geq 2\tau_i^l \} - 2\tau_i^l
   \end{equation}
   
   \vspace{-1em}
   
   \begin{equation}\label{eq:rr2_R:dist}
       R(i,j) = \max_{w > j} \{ \lambda(i,w) \mid \lambda(i,w) - r_{w-1}^r \geq 2\tau_j^r \} - 2\tau_j^r
   \end{equation}
   
   Just like in the solution presented to compute Equation~\ref{eq:dist} in $O(n \log \log n)$,
   we utilize two heaps to calculate the functions $L$ and $R$ in $O(\log \log n)$ time
   at each iteration. 

      To compute the function \( L \), we employ \( n_r \) min-heaps storing \( \lambda(w,j) - r_{w-1}^l \) and \( n_r \) max-heaps  
storing \( \lambda(w,j) \), one for each \( j \in N_r \). These heaps perform the same operations outlined in the proof of  
Theorem~\ref{theom:dist_1}, independently for each \( j \). A similar process is used to compute \( R \).

   \begin{theorem}
        The recurrence relation $\lambda(1,1)$, given by Equation~\ref{eq:rr2:dist}, can be computed in $O(n^2 \log \log n)$.
   \end{theorem}
   
   \begin{proof}
       
   Since each state of $L$ and $R$ can be computed through heaps described in \citet{thorup2000ram} with time 
   complexity of $O(\log \log n)$, and we have $n_r \cdot n_l$ states, the complexity of this solution amounts to
    $O(n^2 \log \log n)$. A similar proof to that seen in Theorem~\ref{theorem:n_squared} can be
     conducted in this scenario as well, but is omitted here for brevity.
   \end{proof}

   Analogous commented in the Section~\ref{sec:improve_nlog}, the complexity depends on the choice of the auxiliary heap. If we use 
   a binary heap, which is easier to implement, the complexity becomes $O(n^2\log n)$.

   \section{Concluding remarks}\label{sec:conclusion}

   In this paper, we addressed the Graphical Traveling Salesman Problem with release dates (GTSP-rd) on paths.  
   Our contributions include the development of algorithms that improve existing solutions. These solutions build  
   on previous recurrence relations and employ dynamic programming for efficient implementation.  
   
   For paths with depots at the extremities, we presented an \( O(n) \) solution for GTSP-rd (time) and an  
   \( O(n \log \log n) \) solution for GTSP-rd (distance). Additionally, for general paths where depots can 
   be located anywhere, we introduced an \( O(n^2) \) solution for GTSP-rd (time) and an \( O(n^2 \log \log n) \) 
   solution for GTSP-rd (distance).
   
   Future work can extend these solution strategies to more complex graph structures, such as subdivided stars, constrained trees,  
   or other graphs studied in the context of GTSP.

\bibliography{lipics-v2021-sample-article}

\end{document}